\documentclass[11pt]{llncs}

\usepackage[T1]{fontenc}

\usepackage{graphicx}
\usepackage[margin=1in]{geometry}

\usepackage{amsmath,amssymb}

\usepackage{booktabs}
\usepackage[ruled]{algorithm2e}
\usepackage{tikz}
\usepackage{xcolor}
\usepackage[backref=page,colorlinks=true]{hyperref}
\author{Arghya Chakraborty\inst{1}\orcidID{0009-0005-5415-8916} \and
Varun Gupta\inst{2}\orcidID{0000-0001-7373-1734}}

\institute{Tata Institute of Fundamental Research, India, \email{arghya314@yahoo.com}\and
University of Utah, Utah 84112, USA \email{varun.gupta@eccles.utah.edu}}
\hypersetup{
pageanchor=true,
bookmarksnumbered,
colorlinks=true,
menubordercolor= [rgb]{0, 0, 1},
linkbordercolor= [rgb]{0, 1, 1},
citecolor = [rgb]{0.7, 0, 0},
linkcolor= [rgb]{0, 0, 0.8}
}

\SetAlFnt{\small}
\SetAlCapFnt{\small}
\SetAlCapNameFnt{\small}
\SetAlCapHSkip{0pt}
\IncMargin{-\parindent}
\allowdisplaybreaks

\newcommand{\OPT}{\mathsf{OPT}}
\newcommand{\E}{\mathbb{E}}
\newcommand{\gain}{\text{gain}}
\DeclareMathOperator{\Var}{Var}
\DeclareMathOperator{\rank}{rank}

\newenvironment{restatement}[1]
  {\par\medskip
   \noindent\textbf{Restatement of #1.}\itshape\ }
  {\par\medskip}

\title{Arrival-Time Incentive Compatibility in Random Order Online Bipartite Matching}

\date{}

\begin{document}

\maketitle
\thispagestyle{empty}

\begin{abstract}
A common motif in designing online algorithms in stochastic or random order settings is an initial exploration (or sampling) phase to learn useful information about the instance (e.g., dual prices, a primal resource allocation plan, or admission thresholds). By construction, exploration leads to unfair outcomes for early arrivals leading to incentives for agents to influence when they show up. The focus of the present paper is on edge-weighted online bipartite matching where one side (servers) is available offline, and the other side (users) arrives online in a random order (secretary model). There is a service reward associated with each (server, user) pair that the central designer wants to maximize. In an elegant generalization of the classic secretary algorithm, \cite{kesselheim2013optimal} propose an optimal $1/e$ competitive algorithm for edge-weighted online bipartite matching in the secretary model. A key ingredient of their algorithm is to effectively discard (that is, use only for learning) the first $\approx n/e$ arrivals, thereby ex-ante disadvantaging earlier arrivals relative to later arrivals.

In this work we initiate the study of competitive algorithms with arrival-time incentive compatibility for random-order online bipartite matching in settings where the users care only about receiving service (matched vs. unmatched) and not which offline resource serves them, while the platform’s objective is to maximize total matching reward. This captures applications such as ride-sharing where the users primarily care about being matched to a ride while the platform internalizes the cost of dispatching a distant driver; dispatching homogeneous service requests to heterogeneous servers (cloud/edge routing); and assigning customer requests to a pool of providers with different flexibility (e.g., English-only vs. bilingual agents). Our main question is: \textit{Is constant-competitive matching possible for incentive-compatible, random-order edge-weighted matching on complete bipartite graphs?}

Motivated by the LP-based treatment of incentive compatibility in the classical secretary problem by \cite{buchbinder2014secretary}, we impose a constraint that the ex ante probability of selection is equalized across all arrival positions. We answer our main question in the affirmative and propose the first constant-competitive algorithm for incentive compatible edge-weighted random-order online matching on complete bipartite graphs. The competitive ratio of our algorithm is parameterized by the imbalance factor $k := n/m$ -- where $n$ and $m$ are the numbers of online and offline nodes, respectively, and $k$ is a positive integer. In particular, we obtain a competitive guarantee of the form $c_k - O(1/\sqrt{m})$ where $c_1 \approx 0.162$ and $c_k \to 0.02308\ldots$ as $k \to \infty$. We also present algorithms with strictly improved competitive ratio of $\approx 0.07 + O(1/m)$ for the binary-weighted case ($0$-$1$ rewards).

\end{abstract}

\setcounter{page}{1}

\section{Introduction}

Many marketplace and resource-allocation systems operate online: requests
arrive sequentially, decisions are irrevocable, and the algorithm has only
partial information about the future. Examples include ride-hailing and
delivery platforms, online advertising, admission control, and cloud/edge
routing. A canonical abstraction of such systems is \emph{online bipartite matching}: offline
vertices represent fixed resources, online vertices represent arriving users,
and each user must be matched immediately to an available resource or rejected.

A key modeling choice is the arrival process. Under adversarial arrivals,
competitive guarantees are often pessimistic, whereas the random-order
(secretary) model permits stronger algorithms. Many such algorithms use an
initial \emph{exploration} phase to learn prices, thresholds, or a primal
allocation plan, and then exploit this information on later arrivals. This
is effective algorithmically but creates temporal disparities: early arrivals
are treated differently from later ones.

The classical secretary problem already illustrates the issue. The optimal
$1/e$-success algorithm samples and rejects an initial prefix, giving agents
a clear incentive to avoid early positions. Buchbinder et al.~\cite{buchbinder2014secretary}
addressed this by requiring every arrival position to have the same ex-ante
selection probability, while still obtaining a constant probability of hiring
the best candidate. In online bipartite matching, however, existing
random-order algorithms are also time-asymmetric. In particular, the optimal
$1/e$-competitive edge-weighted matching algorithm of
Kesselheim et al.~\cite{kesselheim2013optimal} relies on using an initial prefix only for
learning. Such algorithms can induce users to delay, undermining the
random-order assumption they rely on.

\paragraph{Model and incentive constraint.}
We study edge-weighted online bipartite matching in the random-order model.
Offline vertices are heterogeneous servers, online vertices are users, and
the platform seeks to maximize total matching weight. Users care only about
receiving service, not about which server they are assigned to. We extend the
arrival-time incentive-compatibility constraint of
Buchbinder et al.~\cite{buchbinder2014secretary}: under a uniform prior over the
user's type and a uniformly random arrival order, the ex-ante probability of
being matched must be identical for every arrival position $t\in[n]$. Thus no
user can improve their ex-ante probability of service by manipulating their arrival time.

\paragraph{Central question.}
\emph{Is constant-competitive edge-weighted online bipartite matching
possible in the random-order model under arrival-time incentive
compatibility?}
\paragraph{Results.}
Our main result can be interpreted as a separation between complete and incomplete bipartite graphs in regard to the answer to the question above. We first prove that constant competitiveness is not possible under arrival-time incentive compatibility if the bipartite graph has \textit{infeasible edges}, i.e., is not necessarily complete (Section~\ref{sec:lb-forbidden}). We therefore study \emph{complete} bipartite graphs, which better capture applications where every user can in principle be served by any server but the platform internalizes heterogeneous service costs, in the rest of the paper. For complete graphs, we give a surprising affirmative answer to the above question by presenting the first constant-competitive algorithms for arrival-time incentive compatible online bipartite matching in the secretary model. The separation between complete and incomplete bipartite graphs is interesting because adding $0$-weight edges to a graph does not alter the value of the maximum-weight matching, but such edges help us maintain incentive compatibility. Our competitive guarantees assume \(n=km\) for a positive integer \(k\), where \(m\) and \(n\) are the numbers of offline and online vertices, respectively. The following bullets summarize our results:
\begin{itemize}
    \item In bipartite graphs with infeasible match pairs (that is, not complete bipartite graphs), we prove that no incentive compatible online algorithm can guarantee a competitive ratio better than $1/n$, even in the unweighted setting where all feasible edges have weight 1 (Section~\ref{sec:lb-forbidden}).
    \item In the balanced case \(m=n\), we give an incentive-compatible algorithm achieving asymptotic competitive ratio \(\approx 0.162\) (Section~\ref{sec:balanced_wt}).
    \item For \(n=km\) with any positive integer \(k\), we give an incentive-compatible algorithm with competitive ratio \(c_k - O(1/\sqrt{m})\), where \(c_1\approx 0.162\) and \(c_k\to 0.02308\ldots\) as \(k\to\infty\) (Section~\ref{sec:edge_wt}).
    \item In the binary-weighted (edge weights in $\{0,1\}$) case, we provide an improved competitive ratio of $\approx 0.071$ (Section~\ref{sec:unweighted}).
\end{itemize}

\paragraph{Techniques and obstacles.}
Our edge-weighted algorithm starts from the random-order matching framework of Kesselheim et al.~\cite{kesselheim2013optimal}: as arrivals are revealed, the algorithm uses an optimal matching on the revealed instance to define high-value ``tentative'' edges. The obstacle is that this framework naturally treats online arrivals at different times asymmetrically. We enforce this symmetry through a block-based construction. The sequence is divided into equal-sized blocks, and within each block we run the incentive-compatible secretary algorithm of Buchbinder et al.~\cite{buchbinder2014secretary}, so that at most one online vertex is matched from each block.

Making this compatible with the random-order matching analysis requires two new ingredients. First, the secretary subroutine must see the tentative-edge weights inside a block in a random order, so we modify the Kesselheim et al.~\cite{kesselheim2013optimal} tentative-edge construction by evaluating each candidate against the previously revealed blocks together with that candidate alone. Second, we need to show that the offline endpoint of a selected tentative edge is still available with constant probability. This is not a direct consequence of the original analysis: in our imbalanced setting, a block contains several online vertices, and hence a fixed offline vertex can appear as a tentative endpoint multiple times within the same block. Our main analytical novelty is an \emph{amortized availability} argument that averages over both the random partition into blocks and the random order inside each block.
For binary weights, we use a separate greedy block algorithm and a different amortized accounting proof to obtain a better constant. We elaborate on these two ingredients in Section~\ref{sec:edge_wt}.

\paragraph{Organization.}
Section~\ref{sec:related} reviews related work. Section~\ref{sec:prelim} introduces the model and incentive-compat\-ibility definition. In Section~\ref{sec:lb-forbidden} we justify studying complete bipartite graphs by proving that no constant competitive ratio is possible if the bipartite graph has forbidden pairs of matches. In Section~\ref{sec:balanced_wt} we present an algorithm for incentive compatible bipartite matching in the balanced setting to highlight the main ideas, and then treat the general (imbalanced) edge-weighted setting in Section~\ref{sec:edge_wt}. In Section~\ref{sec:unweighted} we present a slightly different algorithm and a tighter analysis for the binary-weighted case (edge weights in $\{0,1\}$). We conclude in Section~\ref{sec:conclusion} with some open questions.

\section{Related Work}
\label{sec:related}

\textbf{Online bipartite matching. }
Online bipartite matching was introduced by Karp et al.~\cite{KarpVV90}, whose \textsc{Ranking} algorithm achieves the optimal $1-1/e$ competitive ratio for unweighted matching under adversarial arrivals. Subsequent work refined the analysis and primal--dual/randomized dual-fitting viewpoints; see Devanur et al.~\cite{DevanurJK13}. The vertex-weighted case was addressed by Aggarwal et al.~\cite{AggarwalGKM11}. In the fully edge-weighted setting, greedy is $1/2$-competitive under free disposal, and improving beyond $1/2$ requires more delicate correlation/control arguments; see Fahrbach et al.~\cite{FahrbachHTZ22JACM}. A parallel line, motivated by ad allocation and budgets, includes the AdWords framework of Mehta et al.~\cite{MehtaSVV07} and the survey of Mehta~\cite{MehtaSurvey13}.

\noindent \textbf{Random-order and stochastic arrivals.}
Our closest algorithmic backdrop is online matching in the random-order (secretary) model, where edge weights are adversarial but online vertices arrive in a uniformly random order. Early work connected random-input models to ad allocation \cite{GoelMehta08}; for unweighted matching, Mahdian et al.~\cite{MahdianYan11} obtained improved guarantees via factor-revealing LPs. For weighted matching, Kesselheim et al.~\cite{kesselheim2013optimal} gave an optimal random-order algorithm, and Korula et al.~\cite{KorulaPal09} developed the broader ``secretary on graphs'' perspective. \cite{kesselheim2014primal} studied online packing Linear Programs in the random order setting, and proposed an algorithm based on linear scaling of resource capacities. In terms of the ``capacity ratio'' (ratio of resource capacity to its maximum demand by any single arrival), their algorithm achieves a more favorable guarantee compared to dual price based algorithms. We study the same random-order edge-weighted matching problem, but impose an ex-ante time-symmetry / incentive-compatibility constraint that rules out algorithms whose learning phase systematically disadvantages early arrivals. While the primal-based algorithm of \cite{kesselheim2014primal} also avoids an explicit learning phase, it is not clear how their algorithm can be adapted to guarantee the arrival-time IC property. A stronger stochastic model assumes known i.i.d. arrivals, enabling ratios above the adversarial $1-1/e$ barrier in several settings \cite{FeldmanMMM09,ManshadiGS12,JailletLu14}; for stochastic weighted matching see Haeupler et al.~\cite{HaeuplerMZ11,brubach2020online}, and for unknown distributions see Karande et al.~\cite{KarandeMT11}.

\noindent \textbf{Fairness, incentives, and exploration.}
Fairness in irrevocable online decisions has been studied for online selection \cite{CorreaCDN21} and online matching, including group-level service guarantees \cite{MaXX20}, long-run fairness among dynamic agents
\cite{MaXu24}, and fair-division notions such as envy-freeness, proportionality, and maximin share in online matching with indivisible items \cite{HosseiniHIS23,HajiaghayiJSSS24}. These works typically constrain
fairness across agents, classes, or groups; our constraint is instead temporal, requiring equal matching probability across arrival positions. This is also related to incentive issues in sequential exploration, where
platforms use information or recommendations to induce exploration in bandits \cite{FrazierKKK14,MansourSS15}, persuasion/disclosure models \cite{KremerMP14}, and matching markets \cite{NgoPV24}. Our approach is complementary: rather than incentivizing exploration through information, we
build time-symmetry directly into the online matching rule while retaining constant competitive guarantees.

\section{Preliminaries}
\label{sec:prelim}
An instance is a weighted complete bipartite graph
$G=(L,R,w)$, where $L=\{u_1,\ldots,u_m\}$ is the offline
side, $R=\{v_1,\ldots,v_n\}$ is the online side, and
$w:L\times R\to\mathbb{R}_{\ge 0}$ gives the edge rewards.
The algorithm knows $L$ and $n=|R|$ in advance. The vertices of
$R$ arrive one at a time in a uniformly random order; let $r_t$
denote the vertex arriving at time $t\in[n]$. Upon arrival of
$r_t$, the algorithm observes all incident weights
$\{w(u,r_t):u\in L\}$, but not the weights of future arrivals, and
must irrevocably either match $r_t$ to an unmatched offline vertex
or leave it unmatched. We use subscripts $i$ and $j$ for offline
vertices $u_i$ and online vertices $v_j$, respectively.

Among maximum-weight matchings, we choose one of maximum cardinality, subject
to any stated cardinality constraint, and break remaining ties using a fixed
lexicographic order on edge sets, independent of arrival order and consistent
across subinstances.

The objective is to maximize the reward of the resulting matching:
the number of matched weight-$1$ edges in the binary-weighted case,
and the total matched edge weight in the general edge-weighted case.
We measure performance by competitive ratio.

\begin{definition}[Competitive ratio of a maximization problem]
\label{def:CR}
    Given a randomized algorithm $A$, the competitive ratio of $A$ is given by
    \[
        \gamma_A := \min_{\text{Graph } G} \frac{\mathbb{E}[A(G)]}{\OPT(G)}.
    \]
\end{definition}

    Here, $\mathbb{E}[A(G)]$ and $\OPT(G)$ denote the expected value of the objective achieved by the algorithm $A$ and the optimal offline algorithm, respectively, on instance $G$. The expectation is taken over the randomness of the algorithm and the randomness of the arrival order of $G$. To suppress dependence on $G$, we will use $\OPT(G)$ and $\OPT$ interchangeably. Next, we define our notion of incentive compatibility.

\begin{definition}[Arrival-time incentive compatible (IC) matching]\label{def:IC}
Fix a bipartite graph $G$ with $|R|=n$ online vertices. Let $\pi$ be a uniformly random
permutation of $R$, and let $r_t := v_{\pi(t)}$ denote the (random) online vertex arriving at time $t\in[n]$.
A (possibly randomized) online matching algorithm $A$ is \emph{arrival-time incentive compatible (IC)}
if for every instance $G$ and every pair of positions $s,t\in[n]$,
\[
\Pr_{\pi,\;A}\!\big[\, r_s \text{ is matched}\,\big]
=
\Pr_{\pi,\;A}\!\big[\, r_t \text{ is matched}\,\big],
\]
where the probability is over the random arrival order $\pi$ and the internal randomness of $A$.
\end{definition}
Our notion of IC is subtle, but is analogous to the IC secretary problem of \cite{buchbinder2014secretary}. The secretary candidates in \cite{buchbinder2014secretary} are Bayesian agents. Each agent has an identical common uniform prior over their ranks (that is, each agent believes themselves to be the best candidate with probability $1/n$, second best with probability $1/n$, and so on), and evaluates their assigned interview position before knowing their realized rank/type. Under this interpretation, equalizing the ex-ante selection probability across positions removes incentives to manipulate arrival time. The algorithms of \cite{buchbinder2014secretary} do not guarantee IC if agents know their realized rank when evaluating whether to deviate from their interview position (and we believe that constant competitiveness is impossible for that setting). Similarly, we assume agents know the offline graph, but have a uniform prior over their \textit{type}. The type of an agent is defined as the node in $R$ they are mapped to. Definition~\ref{def:IC} equalizes the \textit{ex-ante} service probability across arrival positions under this uniform prior. It \textbf{does not} offer the ex-post guarantee
\[
\Pr_{\pi,\;A}\!\big[\, v_{j_1} \text{ is matched}\, \big]
=
\Pr_{\pi,\;A}\!\big[\, v_{j_2} \text{ is matched}\, \big]
\]
for two different types, \textbf{nor} does it require the stronger type-conditioned arrival time IC property
\[
\Pr_{\pi,\;A}\!\big[\, v \text{ is matched}\, \mid v \text{ arrives at } s \big]
=
\Pr_{\pi,\;A}\!\big[\, v \text{ is matched}\, \mid v \text{ arrives at } t \big]
\]
for every fixed vertex $v \in R$ of the offline graph.

\section{An upper bound in the presence of infeasible edges}\label{sec:lb-forbidden}
The following theorem states that if the bipartite matching graph is not complete, that is, some pairs of matches are infeasible, then no IC algorithm (even randomized) can be constant competitive. This impossibility result motivates our assumption of complete bipartite graphs in the sequel.

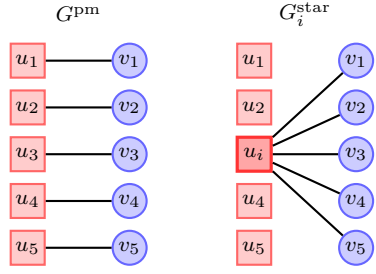
\begin{figure}[t]
    \centering
\begin{tikzpicture}[
  scale=1, transform shape,
  every node/.style={font=\scriptsize},
  online/.style  ={circle, draw=blue!60,  fill=blue!20,  thick,
                   minimum size=4.5mm, inner sep=0.5pt},
  offline/.style ={rectangle, draw=red!60, fill=red!20, thick,
                   minimum size=4.5mm, inner sep=0.5pt},
  offlineCenter/.style={rectangle, draw=red!80, fill=red!35, very thick,
                        minimum size=4.5mm, inner sep=0.5pt},
  edge/.style={thick}
]
\def\n{5}
\def\dy{0.62}
\def\xU{0}
\def\xV{1.35}
\def\shift{3.0}
\def\istar{3}

\node at ({(\xU+\xV)/2},0) {$G^{\mathrm{pm}}$};

\foreach \j in {1,...,\n}{
  \node[offline] (pmu\j) at (\xU, -\j*\dy) {$u_{\j}$};
  \node[online]  (pmv\j) at (\xV, -\j*\dy) {$v_{\j}$};
  \draw[edge] (pmu\j) -- (pmv\j);
}

\node at ({\shift+(\xU+\xV)/2},0) {$G_{i}^{\mathrm{star}}$};

\foreach \j in {1,...,\n}{
  \ifnum\j=\istar
    \node[offlineCenter] (staru\j) at (\shift+\xU, -\j*\dy) {$u_i$};
  \else
    \node[offline] (staru\j) at (\shift+\xU, -\j*\dy) {$u_{\j}$};
  \fi
}

\foreach \j in {1,...,\n}{
  \node[online] (starv\j) at (\shift+\xV, -\j*\dy) {$v_{\j}$};
  \draw[edge] (staru\istar) -- (starv\j);
}

\end{tikzpicture}
    \vspace{-0.6em}
    \caption{Lower-bound instances: a perfect matching $G^{\mathrm{pm}}$ and a star $G_i^{\mathrm{star}}$ centered at $u_i$.}
    \label{fig:forbidden-lb}
\end{figure}

\begin{theorem}[No constant competitive ratio with infeasible edges]\label{lem:forbidden-lb}
Fix $n\ge 1$. Consider unweighted online bipartite matching where each arriving online vertex reveals its feasible incident edges on arrival, and the goal is to maximize the size (cardinality) of the matching. Let ${A}$ be any (randomized) algorithm that is arrival-time
incentive compatible in the sense of Definition~\ref{def:IC}.
Then there exists a balanced bipartite graph $G$ with $|L|=|R|=n$ such that $\OPT(G) \ge n \cdot \E[|A(G)|]$.
In particular, no incentive-compatible algorithm can achieve a constant competitive ratio on
general (non-complete) bipartite graphs even in the unweighted setting.
\end{theorem}
We defer the full proof to Appendix~\ref{sec:lb-forbidden-proof}, but the main idea is to look at the `perfect matching' and `star' instances shown in Figure~\ref{fig:forbidden-lb}. IC on any star bounds the first-arrival match probability by $1/n$. At the first arrival, a vertex with neighbor $u_i$ is indistinguishable from the same vertex on the star centered at $u_i$. Averaging over labelings of the perfect-matching instance therefore gives an instance $G^{\mathrm{pm}}$ whose first-arrival match probability is at most $1/n$. By IC, every arrival position on this instance has match probability at most $1/n$, so the expected reward is at most $1$, while $\OPT(G^{\mathrm{pm}})=n$.

\section{Balanced Complete Graphs}
\label{sec:balanced_wt}

In this section we consider the special case where the cardinalities of the offline and online sides of the bipartite graph are equal, deferring the general case to Section~\ref{sec:edge_wt} (which subsumes the current section). The balanced setting allows isolating some of the ideas used for the general case in a simpler setting, and also allows us to present some of the relevant ideas from \cite{kesselheim2013optimal}.
Our main result of this Section is the following theorem.
\begin{theorem}\label{thm:balanced}
    Given a weighted bipartite graph with $|L|=|R|=n$, Algorithm~\ref{alg:balanced} produces an incentive compatible matching with an asymptotic competitive ratio of $\approx 0.162$.
\end{theorem}
The key idea is to guarantee IC by ensuring that the algorithm always outputs a perfect matching. Thus, the problem reduces to finding a perfect matching in the random-order arrival model with a good competitive ratio. Our algorithm extends the $1/e$-competitive online edge-weighted, but not IC, matching algorithm of \cite{kesselheim2013optimal}. Recall that $r_t$ is the $t$th arriving online vertex. The algorithm first matches the first $T-1$ online vertices uniformly at random (we optimize $T$ later). For each subsequent arrival $r_t$, it computes an optimal offline matching $M_t^\star$ of cardinality $t$ in the graph revealed so far, between $L$ and ${r_1,\ldots,r_t}$. Let $e_t=(\ell_t,r_t)$ be the edge of $M_t^\star$ incident to $r_t$, which exists since $|L|=|R|$. If $\ell_t$ is unmatched, we match $\ell_t$ to $r_t$; otherwise, we match $r_t$ to a uniformly random unmatched offline vertex, denoted by $f_t=(q_t,r_t)$. To implement this uniform sampling, we pre-sample a permutation $\rho$ of $L$ before arrivals begin and, whenever a random unmatched offline vertex is needed, choose the highest-ranked available vertex under $\rho$. We also write $f_t$ for the random edge used when $t<T$, and set $f_t=\emptyset$ whenever no random match is made.

\begin{algorithm}[ht]
\caption{Online Matching on a balanced graph}
\KwIn{Bipartite graph $(L,R)$ with $|L|=|R|$, prefix length $T$}
Sample a random permutation $\rho$ of $L$\;

\For{$t = 1$ \KwTo $T-1$}{
    Match $r_t$ to the highest-ranked available vertex in $L$ according to $\rho$\;
}

\For{$t = T$ \KwTo $|R|$}{
    Compute an optimal offline matching $M_t^\star$ between $L$ and $\{r_1, \ldots, r_t\}$\;
    Let $e_t = (\ell_t, r_t)$ be the edge incident to $r_t$ in $M^\star_t$ (called a `\textit{tentative}' edge)\;

    \eIf{$\ell_t$ is unmatched}{
        Match $\ell_t$ to $r_t$\;
    }{
        Choose $f_t=(q_t,r_t)$ by selecting $q_t$ as the highest-ranked available offline vertex according to $\rho$\;
        Match $r_t$ to $q_t$\;
    }
}\label{alg:balanced}
\end{algorithm}

\begin{proof}[Proof of Theorem \ref{thm:balanced}]
 The analysis follows from Lemmas \ref{lem:edge_avail_kesselheim}-\ref{lemma:low_prob_ran_match}.  Lemma~\ref{lem:edge_avail_kesselheim} shows that the ``tentative'' edge $e_t$ captures a reasonable fraction of the weight of $\OPT$ in expectation. However, the online algorithm might be unable to match $e_t$ if the offline vertex $\ell_t \in e_t$ was already matched. Lemmas~\ref{lemma:low_prob_greedy_match} and \ref{lemma:low_prob_ran_match} show that the offline vertex $\ell_t$, and in fact every $u \in L$, has at least an $\Omega(1)$ probability of being available at time step $t$. Lemma~\ref{lemma:low_prob_greedy_match} shows that the probability that an arbitrary offline vertex $u$ is not picked as a `tentative' edge by time $t$ is uniformly lower bounded (uniform when conditioning over the sequence of future online arrivals). The proof of this lemma carries over directly from the work of \cite{kesselheim2013optimal} but is provided here for completeness. Finally, unlike the algorithm in \cite{kesselheim2013optimal}, even though an offline vertex $u$ might not get matched as a tentative edge, it might still end up getting matched as part of an edge $f_s$ during random sampling. Lemma~\ref{lemma:low_prob_ran_match} upper bounds this probability.

\begin{lemma}
\label{lem:edge_avail_kesselheim}
The expected weight of the ``tentative edge'' $e_t$ for any $T \leq t \leq n$ is at least $\OPT/n$.
\end{lemma}
\begin{proof} This lemma directly follows from Kesselheim et al.~\cite{kesselheim2013optimal}. Let $V_t \subseteq R$ be the \textit{unordered} set of the vertices $\{r_1,\ldots, r_t\}$ revealed so far. The expected total weight of the matching $M_t^\star$ between $L$ and $V_t$ is at least $\OPT\cdot t/n$ since $V_t$ is a random subset of size $t$ of $R$. Conditioned on $V_t$, $r_t$ is a uniform sample from $V_t$ and therefore contributes in expectation exactly $w(M_t^\star)/t$.
\end{proof}

\begin{lemma}\label{lemma:low_prob_greedy_match}
    Fix any arbitrary offline vertex $u \in L$, and an arbitrary time $ t \in \{T-1,\ldots, n\}$. Then,
    \[
        \mathbb{P}\left[ u \notin e_s \text{ for all }  T \leq s \leq t  \mid r_{t+1}, \ldots, r_n\right] \geq \frac{T-1}{t}.
    \]
\end{lemma}
\begin{proof} The proof proceeds by induction, the claim being vacuously true for $t=T-1$. Assume the result is true up to $t-1$, and we aim to prove it for $t$. Conditioning on the sequence of future arrivals $r_{t+1}, \ldots, r_n$ fixes the subset $V_t \subseteq R$ of the first $t$ arrivals, which is sufficient to fix the matching $M_t^\star$, but leaves the arrival order of the first $t$ online vertices as random.
If $u$ is unmatched in $M_t^\star$, the claim for time $t$ is immediate.
Otherwise, let $(u,v^\star)$ be its incident edge in $M_t^\star$, that is, the offline vertex $u$ is potentially part of a `tentative' matched edge. If $r_t \neq v^\star$, then $u$ is not an endpoint of the tentative edge at time $t$, and this event happens with probability at least $(t-1)/t$.
Hence,
\begin{align*}
 &     \mathbb{P}\left[ u \notin e_s \text{ for all }  T \leq s \leq t  \mid r_{t+1}, \ldots, r_n\right] \\
 &= \sum_{v \in V_t} \mathbf{1}_{\{v \neq v^\star \}} \mathbb{P}\left[ r_t = v\right] \mathbb{P}\left[ u \notin e_s \text{ for all }  T \leq s \leq t-1  \mid r_t = v, r_{t+1}, \ldots, r_n\right] \\
 & \geq (t-1) \cdot \frac{1}{t}\cdot \frac{T-1}{t-1} \ = \ \frac{T-1}{t}.
\end{align*}
\end{proof}

\begin{lemma}\label{lemma:low_prob_ran_match}
    Fix an arbitrary offline vertex $u \in L$, and an arbitrary time $t \in \{T-1, \ldots, n\}$. Then,
    \[
        \mathbb{P}\left[ u \notin f_s \mbox{ for all } 1 \leq s \leq t  \mid u \notin e_s \mbox{ for }  T \leq s \leq t , \left( r_{t+1}, \ldots, r_n \right) \right] \geq 1-\frac{t}{n}.
    \]
\end{lemma}
\begin{proof}
Recall that the same uniformly random permutation $\rho$ is used for all initial and fallback matches. Whenever a random match is needed, the algorithm picks the highest-ranked available offline vertex according to $\rho$. By time $t$ at most $t$ offline vertices have been matched. Thus, if $u$ is not among the first $t$ vertices in $\rho$ and is not a tentative endpoint at any time $s\in\{T,\ldots,t\}$, then it is still available at the end of time $t$. The permutation $\rho$ is independent of the arrival order and of the tentative edges, so the conditional probability that $u$ is outside these first $t$ positions is $1-t/n$.
\end{proof}

We now put together the above Lemmas. Define $\alpha:=\frac{T}{n}$. Let $w_t$ denote the weight of the tentative edge at time $t$. Conditioning on $r_t,\ldots,r_n$ fixes $w_t$ and $\ell_t$. Applying the preceding lemmas at time $t-1$ and weakening the resulting availability bound gives the following lower bound on the overall expected weight:
\begin{align*}
    &\sum_{t=T}^n \mathbb{E}\left[ w_t \cdot \mathbf{1}_{\{ \ell_t \text{ available}\}} \right] \\*
    & =  \sum_{t=T}^n \mathbb{E}\left[ \mathbb{E} \left[ w_t \cdot \mathbf{1}_{\{ \ell_t \text{ available}\}} \mid r_t, \ldots, r_n \right] \right] \\
    & \geq \sum_{t=T}^n \mathbb{E}\left[ \mathbb{E} \left[ w_t \mid r_t, \ldots, r_n \right] \frac{T-1}{t}\cdot \left( 1 - \frac{t}{n}\right) \right] \\
    & = \sum_{t=T}^n \mathbb{E}\left[ w_t \right] \frac{T-1}{t}\cdot \left( 1 - \frac{t}{n}\right) \geq \  \frac{\OPT}{n} \sum_{t=T}^n \frac{T-1}{t}\cdot \left( 1 - \frac{t}{n}\right) \\
    & = \OPT \left( \alpha \int_{x=\alpha}^1 \frac{1-x}{x} dx + O(1/n)\right) \ = \ \OPT \left( \alpha(\log \frac{1}{\alpha} + \alpha - 1) + O(1/n)\right).
\end{align*}
The final expression is maximized for $\alpha^\star \approx 0.203$ giving a competitive ratio of $\approx 0.162 - O(1/n)$.
\end{proof}

\section{Online Matching for Edge-Weighted Complete Graphs}
\label{sec:edge_wt}

We now treat the general imbalanced case, where $m=|L|$ and $n=|R|$ need not be
equal. The balanced
algorithm of Section~\ref{sec:balanced_wt} guaranteed IC
by producing a perfect matching, but this is impossible when $n\gg m$. We
instead partition the arrival sequence into blocks and enforce
IC locally within each block. We describe and analyze the algorithm for
$n=km$, where the blocks have equal size and the same block-level match
probability. For arbitrary $n$, Appendix~\ref{app:padding} gives a padding
construction that preserves arrival-time IC for the original online vertices.
Our main contribution is Algorithm~\ref{alg:ic_edge_wt}, and the main result is summarized in Theorem~\ref{thm:ic_edge_wt} below.

Let $p_k$ denote the probability that a Binomial random variable with $k$ trials and mean 1 takes a value of 1; $q_k$ be the probability that the IC secretary
algorithm of Buchbinder et al.~\cite{buchbinder2014secretary} selects the best of $k$
candidates under its optimal parameter choice, and let $y_k$ be the
corresponding probability that it selects \textit{some} (not necessarily the best) candidate. Define
\[
f(\alpha;\gamma)
=
\alpha\left(-\log \alpha-\gamma(1-\alpha)\right),
\qquad
\alpha_k\in\arg\max_{\alpha\in[0,1]} f(\alpha;y_k),
\qquad
c_k=p_k\cdot q_k\cdot f(\alpha_k;y_k)
.
\]

\begin{theorem}
\label{thm:ic_edge_wt}
\label{thm:ic_edge_wt2}
Algorithm~\ref{alg:ic_edge_wt}, with the padding construction of
Appendix~\ref{app:padding} when necessary, is arrival-time IC under
random-order arrivals for arbitrary $m,n\ge 1$. For $|L|=m$ and $|R|=km$,
where $k$ is a positive integer, its competitive ratio is at least
\[
c_k-O\!\left(1/\sqrt{m}\right).
\]
The values of $p_k, q_k, y_k, c_k$ for small $k$ as well as their asymptotic limits as $k\to \infty$ are:
\[
\begin{array}{c|cccc}
\hline
k & 1 & 2 & 3 & k\to\infty\\
\hline
p_k & 1 & 0.5 & 0.444 & 1/e\\
q_k & 1 & 0.5 & 0.444 & 1-1/\sqrt{2}\\
y_k & 1 & 1 & 1 & 1/\sqrt{2}\\
c_k & 0.162 & 0.041 & 0.032 & 0.023\\
\hline
\end{array}
\]
\end{theorem}

\subsection{Algorithm}

We partition the online sequence into $m$ contiguous blocks
$R_1,\ldots,R_m$, each of size $k$, and write
$r_{t,s}$ for the $s$th arrival in block $R_t$. As in Section~\ref{sec:balanced_wt}, we match at
most one online vertex per block. For each of the first $T-1$ blocks, where
$T=\lceil \alpha_k m\rceil$, we independently with probability $y_k$ select a uniformly random online vertex within the block, and match it to a random
available offline vertex.

For each later block $R_t$, we treat its $k$ online vertices as candidates in an IC secretary instance. The value assigned to candidate $r_{t,s}$ is computed from the \emph{maximum-weight cardinality-constrained} matching of size $t$ between $L$ and $R_1\cup\cdots\cup R_{t-1}\cup\{r_{t,s}\}$. If $r_{t,s}$ is matched in this offline matching, we let $e_{t,s}=(\ell_{t,s},r_{t,s})$ denote its incident ``tentative'' edge and assign secretary value $w(e_{t,s})$; otherwise, its value is $0$. The key point is that each candidate in the current block is evaluated separately, using the previous blocks together with that candidate alone. Hence, conditioned on $R_t$, the values presented to the secretary subroutine form a fixed multiset arriving in uniformly random order. We break ties using an independent random priority sampled before the block is processed. This is also the first main point at which our algorithm departs from Kesselheim et al.~\cite{kesselheim2013optimal}: the more direct extension of their approach, in which the tentative value of a candidate is computed using all online vertices observed so far, does not preserve this property because the value assigned to a candidate can depend on which other vertices from the same block arrived earlier. Under that construction, the secretary subroutine therefore no longer sees a fixed multiset in random order. Finally, as in Section~\ref{sec:balanced_wt}, if the tentative edge of the candidate selected by the IC secretary algorithm is unavailable, we match the selected online vertex to a uniformly random available offline vertex.

\begin{algorithm}[t]
\caption{Block-Based Online Matching Algorithm}
\label{alg:ic_edge_wt}
\KwIn{Offline vertices $L$, online blocks $R_1,\dots,R_m$, threshold $T$,
block match probability $y_k$}

Sample a random permutation $\rho$ of $L$\;

\For{$t=1$ \KwTo $T-1$}{
    With probability $y_k$, choose $r\sim \mathrm{Unif}(R_t)$ and match it
    to the highest-ranked available offline vertex according to $\rho$\;
}

\For{$t=T$ \KwTo $m$}{
    Start a fresh IC secretary instance of Algorithm~\ref{alg:ic_sec}
    with parameter $p=\eta_k=y_k/k$\;
    \For{$s=1$ \KwTo $k$}{
        Let $M_{t,s}^\star$ be the maximum-weight cardinality-$t$ matching
        between $L$ and $R_1\cup\cdots\cup R_{t-1}\cup\{r_{t,s}\}$\;
        \eIf{$r_{t,s}$ is matched in $M_{t,s}^\star$}{
            Let $e_{t,s}=(\ell_{t,s},r_{t,s})$ be its tentative edge and
            give value $w(e_{t,s})$ to the secretary instance\;
        }{
            Give value $0$ to the secretary instance and set
            $\ell_{t,s}=\emptyset$\;
        }
        \If{the secretary instance selects $r_{t,s}$}{
            \eIf{$\ell_{t,s}\neq\emptyset$ and $\ell_{t,s}$ is available}{
                Match $r_{t,s}$ to $\ell_{t,s}$\;
            }{
                Match $r_{t,s}$ to the highest-ranked available offline
                vertex according to $\rho$\;
            }
        }
    }
}
\end{algorithm}

The algorithm is IC because every block produces a match
with the same probability $y_k$.  In the initial blocks, this is by
construction. In later blocks, it is the hire probability of the IC secretary
subroutine. Within every block, each position has the same ex-ante selection
probability, and since every selected online vertex is matched using either
its tentative edge or a fallback edge, every arrival position is matched with
probability $y_k/k$.

\subsection{Analysis}

Throughout this analysis, $n=km$. We expose the uniformly random arrival order in two stages: first a uniformly
random partition of $R$ into $m$ blocks of size $k$, followed by a uniformly
random ordering of the vertices inside each block. Let $M^\star$ be an
optimal offline matching.

Lemma~\ref{lemma:ic_algo_gives_good_gain} says that the best tentative edge in a block carries a
constant fraction of the offline optimum in expectation. Its proof, deferred
to Appendix~\ref{sec:proof_supporting_lemmas}, combines Lemma~\ref{lem:edge_avail_kesselheim} with a unit-occupancy calculation for random blocks.

\begin{lemma}[Tentative-edge value]
\label{lemma:ic_algo_gives_good_gain}
For every block $t\in[T,m]$, let
$e_{t,1},\ldots,e_{t,k}$ be the tentative edges defined by
Algorithm~\ref{alg:ic_edge_wt}. Then
\[
\E\!\left[\max_{s\in[k]} w(e_{t,s})\right]
\ge
\frac{p_k}{m}\,w(M^\star)
\left(1-O\!\left(\frac1{\sqrt m}\right)\right).
\]
\end{lemma}

Lemma~\ref{lemma:amortized_availability} (analog of Lemmas~\ref{lemma:low_prob_greedy_match} and \ref{lemma:low_prob_ran_match}) is the primary new analytical ingredient, and our second main point of departure from Kesselheim et al.~\cite{kesselheim2013optimal}. It lower bounds the
probability that a tentative endpoint survives until its block arrives. The availability analysis of \cite{kesselheim2013optimal} relies only on the random order of arrival, and trying to apply it by fixing the partition of online vertices into the blocks and exploiting random arrival order of blocks and random arrival order within blocks does not work. In fact, the result is not true if we fix the partition of online vertices into the blocks. This is because if we only rely on the randomness of arrivals of blocks and of how online nodes arrive within a block, we might have offline vertices that are less likely to be available than others. For instance, the same offline vertex may have the highest edge-weight with several candidates within one block of online vertices. Lemma \ref{lemma:amortized_availability} addresses this availability analysis by \textbf{additionally} amortizing over the random block partition when considering the probability of the vertex being matched as either a tentative edge or a fallback edge.

\begin{lemma}[Amortized availability]
\label{lemma:amortized_availability}
Fix any offline vertex $u\in L$ and block $t\in[T,m]$. Conditioning on the
current and future blocks $(R_t,\ldots,R_m)$, the probability that $u$ is available
at the beginning of block $t$ is at least
\[ \left(\frac{T-1}{t-1}-O\!\left(\frac1m\right)\right)
\left(1-y_k\frac{t-1}{m}\right).
\]
\end{lemma}

\begin{proof}[Proof of Theorem~\ref{thm:ic_edge_wt2}]
It remains only to prove the competitive ratio. Consider a block $t\ge T$.
Conditioned on the candidate values in this block, the IC secretary
subroutine selects the maximum-value candidate with probability $q_k$.
By Lemma~\ref{lemma:ic_algo_gives_good_gain}, the expected value of this
maximum tentative edge is at least
\[
\frac{p_k}{m}w(M^\star)
\left(1-O\!\left(\frac1{\sqrt m}\right)\right).
\]
Conditioning on $R_t,\ldots,R_m$ and on the order and secretary randomness
within block $t$, Lemma~\ref{lemma:amortized_availability} implies that the offline
endpoint of the selected tentative edge is available with probability at least
\[
\left(\frac{T-1}{t-1}-O\!\left(\frac1m\right)\right)
\left(1-y_k\frac{t-1}{m}\right).
\]
Therefore,
\[
\E[\mathrm{ALG}]
\ge
w(M^\star)\frac{p_kq_k}{m}
\left(1-O\!\left(\frac1{\sqrt m}\right)\right)
\sum_{t=T}^m
\left(\frac{T-1}{t-1}-O\!\left(\frac1m\right)\right)
\left(1-y_k\frac{t-1}{m}\right).
\]
Writing $T=\lceil \alpha m\rceil$ and approximating the sum by the
corresponding integral gives
\[
\frac1m\sum_{t=T}^m
\frac{T-1}{t-1}
\left(1-y_k\frac{t-1}{m}\right)
=
\alpha\int_\alpha^1 \frac{1-y_kx}{x}\,dx
\pm O\!\left(\frac1{\sqrt m}\right).
\]
Absorbing the $O(1/m)$ availability correction into the $O(1/\sqrt m)$ error term, we obtain
\[
\E[\mathrm{ALG}]
\ge
w(M^\star)p_kq_k
\left(
\alpha\left(\log\frac1\alpha-y_k(1-\alpha)\right)
-
O\!\left(\frac1{\sqrt m}\right)
\right).
\]
Optimizing over $\alpha$ gives
\[
\E[\mathrm{ALG}]
\ge
w(M^\star)
\left(
p_kq_k f(\alpha_k;y_k)
-
O\!\left(\frac1{\sqrt m}\right)
\right),
\]
which proves the theorem.
\end{proof}

\section{Binary-weighted Complete Graphs}
\label{sec:unweighted}

We next show that binary weights admit a better constant via a different
greedy block-based algorithm. Here edge weights lie in $\{0,1\}$, but the
graph is still complete: weight-$0$ edges are feasible matches. This
distinction is immaterial in classical unweighted matching, but is essential
under arrival-time IC, as shown by the upper bound of
Section~\ref{sec:lb-forbidden}.

Let $q_k$ denote the probability that the IC secretary
subroutine $\mathcal{M}_p$ of Buchbinder et al.~\cite{buchbinder2014secretary} selects the best of $k$
candidates under the optimal parameter $p$. Recall that
$q_k \to 1-1/\sqrt{2}\approx 0.29$.

\begin{theorem}
\label{thm:unweighted}
Let $|L|=m$ and $|R|=km$. There is an arrival-time IC
algorithm for binary-weighted online bipartite matching in the random-order
model with competitive ratio at least
\[
\left(1-\frac1e\right)\frac{m+1}{2m}\cdot \frac{q_k}{1+q_k}.
\]
In particular, as $k\to\infty$, this gives a competitive ratio of
 $\approx 0.071$.
\end{theorem}

\paragraph{Algorithm and proof idea.} The full algorithm and proof are deferred to Appendix~\ref{app:binary}.
The algorithm partitions the arrival sequence into $m$ contiguous blocks of
size $k$ and matches at most one online vertex from each block. In each block,
we run a fresh copy of the IC secretary algorithm of
Buchbinder et al.~\cite{buchbinder2014secretary}. For a block, say block $t$, and for every online vertex $r_{t,s} \in R_t$ for $s\in [k]$, its value for the
secretary instance is its current marginal gain:
\[
\text{gain}(r_{t,s})=
\begin{cases}
1, & \text{if } r_{t,s} \text{ has an available weight-}1\text{ neighbor},\\
0, & \text{otherwise.}
\end{cases}
\]
If the secretary subroutine for block $t$ selects $r_{t,s}$, we match it to an available
weight-$1$ neighbor if $\text{gain}(r_{t,s})=1$, and otherwise to a uniformly random available offline vertex. Since the same IC secretary rule is used in every
block and every selected secretary is matched, the ex-ante match probability
is identical across all positions.

For arbitrary $n$, the padding construction in Appendix~\ref{app:padding}
also preserves arrival-time IC for this algorithm.

The analysis uses a different accounting argument from the edge-weighted
case. Let $M^\star_\pi$ be the best offline matching that uses at most one
online vertex from each block, chosen independently of block order. A random block partition preserves a
$(1-1/e)$ fraction of the optimal binary matching in expectation: marking
the online endpoints used by $M^\star$, the expected number of blocks
containing at least one marked vertex is at least $(1-1/e)w(M^\star)$.
Now fix a block $R_j$ that contributes a weight-$1$ edge
$(\ell_j^\star,r_j^\star)$ to $M^\star_\pi$. If $\ell_j^\star$ is still
available when $R_j$ arrives, then some candidate in the block has gain $1$,
and the secretary subroutine captures such a candidate with probability at
least $q_k$. If $\ell_j^\star$ has already been matched by a weight-$1$ edge,
then the algorithm has already earned value that we can charge to
$\ell_j^\star$; if it was consumed by a weight-$0$ edge, we may lose this
opportunity.

To balance these cases, assign reward $\beta$ to the offline endpoint and
$1-\beta$ to the block for every weight-$1$ edge found by the algorithm. If
\begin{align*}
\alpha_j &:= \Pr\!\left[\ell_j^\star \text{ is already matched via a weight-}1\text{ edge when } R_j \text{ arrives}\right],\\
\gamma_j &:= \Pr\!\left[\ell_j^\star \text{ is already matched via a weight-}0\text{ edge when } R_j \text{ arrives}\right],
\end{align*}
then the expected accounting reward for
$(R_j,\ell_j^\star)$ is at least $\alpha_j\beta +(1-\alpha_j-\gamma_j)q_k(1-\beta)$. We also have $\gamma_j\le (j-1)/m$, since by the time $R_j$ arrives at most $j-1$ offline vertices can have been matched. Choosing $\beta=q_k / (1+q_k) $, the contribution of block $j$ is
at least
\[
\left(1-\frac{j-1}{m}\right)\frac{q_k}{1+q_k}.
\]
Averaging over the uniformly random arrival position of each contributing block
and combining with the $(1-1/e)$ factor proves Theorem~\ref{thm:unweighted}.

\section{Conclusion and Future Work}
\label{sec:conclusion}

We introduced arrival-time incentive compatibility for random-order online
bipartite matching and showed that, despite this symmetry constraint,
constant-factor guarantees are possible for complete edge-weighted graphs.

Several directions remain open. First, the optimal competitive ratio is far
from settled. The IC secretary problem provides a natural asymptotic benchmark, where the
optimal best-candidate success probability converges to
$1-1/\sqrt{2}\approx0.293$. Closing the gap between this benchmark and our
matching guarantees, either through sharper upper bounds or improved
algorithms, is a natural next step. We also expect that some constants can be
improved by tightening the crude availability bounds in
Lemmas~\ref{lemma:low_prob_ran_match} and~\ref{lemma:prob_random_match},
for example via an ODE-based analysis.

Second, exact IC may be stronger than necessary in some
applications. An important direction is to study $\varepsilon$-IC, where no
arrival position can improve a user's matching probability by more than
$\varepsilon$, and to characterize the resulting tradeoff between IC and competitive ratio.

Finally, application-driven variants deserve further study. Metric rewards,
as in ride-sharing or spatial service systems, may allow stronger guarantees.
Another variant separates admission from assignment: the platform must
immediately accept or reject each request, but may defer the final offline
assignment until later. This model is natural for reservation systems and
already admits a ratio of $1$ in the balanced complete-graph case.

\bibliographystyle{splncs04}
\bibliography{sample-bibliography}

\appendix
\renewcommand{\theHsection}{appendix.\Alph{section}}

\section{Proof of Theorem~\ref{lem:forbidden-lb}}\label{sec:lb-forbidden-proof}

\begin{restatement}{Theorem~\ref{lem:forbidden-lb}}\
[No constant competitive ratio with infeasible edges]\
Fix $n\ge 1$. Consider unweighted online bipartite matching in the vertex-arrival model, where
each arriving online vertex reveals its feasible incident edges to offline vertices, and decisions are
immediate and irrevocable. Let ${A}$ be any (randomized) algorithm that is arrival-time
IC in the sense of Definition~\ref{def:IC} (i.e., $\Pr[r_s\text{ matched}]=\Pr[r_t\text{ matched}]$
for all $s,t\in[n]$).
Then there exists a balanced bipartite graph $G$ with $|L|=|R|=n$ such that
\[
\frac{\mathbb{E}[|{A}(G)|]}{\OPT(G)} \le \frac{1}{n}.
\]
In particular, no IC algorithm can achieve a constant competitive ratio on
general (non-complete) unweighted bipartite graphs.
\end{restatement}

\begin{proof}
Let the offline side be $U=\{u_1,\dots,u_n\}$ and the online side be $V=\{v_1,\dots,v_n\}$. For each permutation $\sigma$ of $[n]$, define the perfect-matching instance
\[
G^{\mathrm{pm}}_\sigma := (U,V,E^{\mathrm{pm}}_\sigma),
\qquad
E^{\mathrm{pm}}_\sigma := \{(u_{\sigma(j)},v_j): j\in[n]\}.
\]
Each such instance has optimum $n$; Figure~\ref{fig:forbidden-lb} illustrates the identity permutation. For each $i\in[n]$, define the star instance centered at $u_i$ (see Figure~\ref{fig:forbidden-lb}):
\[
G_{i}^{\mathrm{star}} := (U,V,E_{i}^{\mathrm{star}}),
\qquad
E_{i}^{\mathrm{star}} := \{(u_i,v_j): j\in[n]\}.
\]
Since all edges share the same offline endpoint $u_i$, we have $\OPT(G_i^{\mathrm{star}})=1$.

\paragraph{Step 1: IC bounds the per-time match probability on a star.}
Fix $i\in[n]$. By IC of ${A}$ on $G_{i}^{\mathrm{star}}$, there exists a value $p_i$
such that for every arrival index $t\in[n]$,
\[
\Pr[\text{the online vertex arriving at time $t$ is matched on } G_{i}^{\mathrm{star}}] = p_i.
\]
But ${A}$ can match at most one online vertex on $G_{i}^{\mathrm{star}}$, hence
\[
1 \;\ge\; \mathbb{E}[|{A}(G_{i}^{\mathrm{star}})|]
= \sum_{t=1}^n \Pr[\text{$r_t$ is matched on } G_{i}^{\mathrm{star}}]
= n\,p_i,
\]
so $p_i \le 1/n$. In particular,
\[
\Pr[\text{$r_1$ is matched on } G_{i}^{\mathrm{star}}] \le \frac{1}{n}.
\]

\paragraph{Step 2: Averaging over perfect-matching instances.}
For $i,j\in[n]$, let
\[
a_{ij}:=\Pr[\text{$r_1$ is matched on }G_i^{\mathrm{star}}\mid r_1=v_j].
\]
Since the first arriving vertex is uniform in $V$, Step 1 gives
\[
\frac1n\sum_{j=1}^n a_{ij}=p_i\le\frac1n.
\]
On $G^{\mathrm{pm}}_\sigma$, if the first arrival is $v_j$, it reveals the
single neighbor $u_{\sigma(j)}$. This is the same information state as the
first arrival of $v_j$ on $G_{\sigma(j)}^{\mathrm{star}}$, even if the
algorithm uses vertex labels. Therefore,
\[
\Pr[\text{$r_1$ is matched on }G^{\mathrm{pm}}_\sigma]
=\frac1n\sum_{j=1}^n a_{\sigma(j),j}.
\]
Averaging over a uniformly random permutation $\sigma$ gives
\[
\E_\sigma\!\left[\Pr[\text{$r_1$ is matched on }G^{\mathrm{pm}}_\sigma]\right]
=\frac1{n^2}\sum_{i=1}^n\sum_{j=1}^n a_{ij}
=\frac1n\sum_{i=1}^n p_i
\le\frac1n.
\]
Hence there is a fixed permutation $\sigma$ for which this first-arrival
match probability is at most $1/n$. Fix such a permutation and write
$G^{\mathrm{pm}}:=G^{\mathrm{pm}}_\sigma$.

\paragraph{Step 3: Using IC to bound $|A(G^{\mathrm{pm}})|$.}
By IC of ${A}$ on $G^{\mathrm{pm}}$, the probability that $r_t$ is matched is the same for all $t$,
so for every $t\in[n]$,
\[
\Pr[\text{$r_t$ is matched on } G^{\mathrm{pm}}]
= \Pr[\text{$r_1$ is matched on } G^{\mathrm{pm}}]
\le \frac{1}{n}.
\]
Hence,
\[
\mathbb{E}[|{A}(G^{\mathrm{pm}})|]
= \sum_{t=1}^n \Pr[\text{$r_t$ is matched on } G^{\mathrm{pm}}]
\le n\cdot\frac{1}{n} = 1.
\]
Since $\OPT(G^{\mathrm{pm}})=n$, the competitive ratio on $G^{\mathrm{pm}}$ is at most $1/n$.
\end{proof}

\section{\texorpdfstring
  {Incentive Compatible Secretary Algorithm~\cite{buchbinder2014secretary}}
  {Incentive Compatible Secretary Algorithm}}

\begin{algorithm}[H]
\caption{Incentive Compatible Mechanism $\mathcal{M}_p$}
\KwIn{Number of candidates $k$, parameter $p\in(0,1/k]$}
\KwOut{Selected candidate (at most one)}

No candidate is selected initially\;

\For{$i \gets 1$ \KwTo $k$}{
    \If{a candidate is already selected}{
        \textbf{break}\;
    }

    \If{$1 \le i \le \frac{1}{2p}$}{
        $r \gets \dfrac{i}{\frac{1}{p} - i + 1}$\;
        \If{candidate $i$ is the best so far}{
            Select candidate $i$ with probability $r$\;
        }
    }
    \Else{
        $r \gets \dfrac{i}{\frac{1}{p} - i + 1}$\;
        \If{rank of candidate $i$ is in top $\lfloor r \rfloor$}{
            Select candidate $i$\;
        }
        \ElseIf{rank of candidate $i$ is $\lfloor r \rfloor + 1$}{
            Select candidate $i$ with probability $r - \lfloor r \rfloor$\;
        }
    }
}\label{alg:ic_sec}
\end{algorithm}

\section{Supporting Lemmas for Section~\ref{sec:edge_wt}}
\label{sec:proof_supporting_lemmas}

This appendix contains the technical details deferred from
Section~\ref{sec:edge_wt}. Throughout, we assume that $n=km$ and that $m$
is sufficiently large that $T=\lceil\alpha_k m\rceil\ge3$. The finitely many
smaller values of $m$ can be absorbed into the $O(1/\sqrt m)$ error term in
Theorem~\ref{thm:ic_edge_wt2}.
Arrival-time IC for arbitrary arrival counts is established in
Appendix~\ref{app:padding}.

We expose the random arrival order in two stages. First, the online vertices
are partitioned uniformly at random into $m$ blocks
$R_1,\ldots,R_m$, each of size $k$, and the blocks are ordered uniformly at
random. Second, the vertices inside each block are independently permuted
uniformly at random. This two-stage exposure is equivalent to a uniformly
random permutation of $R$.

\subsection{Occupancy and truncation estimates}

We begin with two elementary estimates used in the proof of the tentative-edge
value lemma.

\begin{lemma}
\label{lem:pkasympt}
For integers $k,n\ge 1$, let $p_{k,n}$ be the probability that, in a uniformly
random matching of $n$ left vertices into $nk$ right vertices, exactly one
vertex from a fixed subset of $k$ right vertices is matched. Then
\[
p_{k,n}
=
\left(1-\frac1k\right)^{k-1}
+O\!\left(\frac1n\right).
\]
In particular,
\[
p_k:=\lim_{n\to\infty}p_{k,n}
=
\Pr\!\left[\mathrm{Bin}\!\left(k,\frac1k\right)=1\right]
=
\left(1-\frac1k\right)^{k-1},
\]
and $p_{k,n}\ge 1/e$.
\end{lemma}

\begin{proof}
A uniformly random matching of $n$ left vertices into $nk$ right vertices
induces a uniformly random $n$-subset of matched right vertices. Thus, if
$X$ is the number of matched vertices in a fixed subset of $k$ right vertices,
then $X$ is hypergeometric and
\[
p_{k,n}
=
\Pr[X=1]
=
k\cdot
\frac{\binom{nk-k}{n-1}}{\binom{nk}{n}}.
\]
Writing the binomial coefficients in falling-factorial form gives
\[
p_{k,n}
=
\prod_{j=0}^{k-2}
\frac{n(k-1)-j}{nk-1-j}.
\]
Each factor is at least $1-1/k$, so
$p_{k,n}\ge (1-1/k)^{k-1}\ge 1/e$. For fixed $k$, each factor equals
\[
\frac{k-1}{k}\left(1+O\!\left(\frac1n\right)\right),
\]
and multiplying the constant number $k-1$ of factors gives
\[
p_{k,n}
=
\left(1-\frac1k\right)^{k-1}
+O\!\left(\frac1n\right).
\]
The binomial interpretation follows from
\[
\Pr\!\left[\mathrm{Bin}\!\left(k,\frac1k\right)=1\right]
=
k\cdot \frac1k \left(1-\frac1k\right)^{k-1}.
\]
\end{proof}

\begin{lemma}
\label{lem:hypergeom_truncation}
Let $X$ be a hypergeometric random variable with mean $\mu=np$. Then
\[
\E[\min(X,\mu)]
\ge
\mu-\frac12\sqrt{np(1-p)}.
\]
\end{lemma}

\begin{proof}
Since
\[
\min(X,\mu)=X-(X-\mu)^+,
\]
we have
\[
\E[\min(X,\mu)]
=
\mu-\E[(X-\mu)^+].
\]
Moreover,
\[
\E[(X-\mu)^+]
=
\frac12\E[|X-\mu|]
\le
\frac12\sqrt{\Var(X)}
\le
\frac12\sqrt{np(1-p)},
\]
where the last inequality uses the standard variance bound for a
hypergeometric random variable by the corresponding binomial variance.
\end{proof}

\subsection{Tentative-edge value}

\begin{proof}[Proof of Lemma~\ref{lemma:ic_algo_gives_good_gain}]
Fix a block index $t\in[T,m]$. Let $\widehat M_t$ be the maximum-weight
cardinality-$t$ matching between $L$ and the online vertices in
$R_1\cup\cdots\cup R_t$.

We first lower bound $\E[w(\widehat M_t)]$. Let $M^\star$ be an optimal
offline matching, and let $N_t$ be the number of online endpoints of
$M^\star$ that lie in the first $t$ blocks. Then $N_t$ is hypergeometric
with mean $t$. By keeping the largest $\min(N_t,t)$ edges of $M^\star$ that
appear in the first $t$ blocks, we obtain a feasible cardinality-$t$ matching
on the revealed instance. Therefore, using Lemma~\ref{lem:hypergeom_truncation},
\[
\E[w(\widehat M_t)]
\ge
\frac{\E[\min(N_t,t)]}{m}\,w(M^\star)
\ge
\frac{t}{m}
\left(1-\frac1{2\sqrt t}\right)w(M^\star).
\]

Now condition on the set $R_1\cup\cdots\cup R_t$, and hence on
$\widehat M_t$. Under the random partition of these $tk$ online vertices into
$t$ blocks of size $k$, the probability that exactly one online endpoint
matched by $\widehat M_t$ lies in the last block $R_t$ is $p_{k,t}$. On this
event, let $r_{t,s}$ be that unique endpoint. Then the matching
$\widehat M_t$ is feasible for matching $L$ vs. $R_1\cup\cdots\cup R_{t-1}\cup\{r_{t,s}\}$
because $\widehat M_t$ uses no other vertex from $R_t$. By our consistent
tie-breaking convention, it is also the matching chosen for this restricted
instance. Hence the tentative edge assigned to $r_{t,s}$ is its edge in
$\widehat M_t$. Since the unique edge of $\widehat M_t$ landing in $R_t$ is,
by symmetry, a uniformly random edge of $\widehat M_t$, we get
\[
\E\!\left[
\max_{s\in[k]} w(e_{t,s}) \,\middle|\, \widehat M_t
\right]
\ge
p_{k,t}\cdot \frac{w(\widehat M_t)}{t}.
\]
Taking expectations and applying the previous bound gives
\[
\E\!\left[\max_{s\in[k]} w(e_{t,s})\right]
\ge
\frac{p_{k,t}}{m}
\left(1-\frac1{2\sqrt t}\right)w(M^\star).
\]
Finally, by Lemma~\ref{lem:pkasympt}, $p_{k,t}=p_k+O(1/t)$. Since
$t\ge T=\Theta(m)$, this yields
\[
\E\!\left[\max_{s\in[k]} w(e_{t,s})\right]
\ge
\frac{p_k}{m}w(M^\star)
\left(1-O\!\left(\frac1{\sqrt m}\right)\right).
\]
\end{proof}

\subsection{Amortized availability}

The next lemma is the one-block amortization step. It is the point at which
we average over the random assignment of online vertices to the current block,
rather than conditioning on a fixed block partition.

\begin{lemma}
\label{lem:one_block_amortization}
Let $|R|=bk$ for some $b\ge 2$, and sample a set $V\subseteq R$ of size $k$
uniformly at random. Fix an offline vertex $u\in L$. For each $v\in V$, let
$M_{V,v}$ be the maximum-weight cardinality-$b$ matching between $L$ and
$R\setminus (V\setminus\{v\})$. Then
\[
\E_V\!\left[
\mathbf{1}\{\exists v\in V : (u,v)\in M_{V,v}\}
\right]
\le
\frac{k}{(b-1)k+1}.
\]
\end{lemma}

\begin{proof}
Fix $V$ and write
\[
\mathbf{1}\{\exists v\in V : (u,v)\in M_{V,v}\}
\le
\sum_{v\in V}\mathbf{1}\{(u,v)\in M_{V,v}\}.
\]
We expose $V$ in two steps: first sample $V_1\subseteq R$ with
$|V_1|=k-1$, and then sample $r$ uniformly from $R\setminus V_1$, setting
$V=V_1\cup\{r\}$. Conditional on $V_1$, the matching
$M_{V_1\cup\{r\},r}$ is a maximum-weight cardinality-$b$ matching between
$L$ and $R\setminus V_1$. In this matching, the fixed offline vertex $u$ is
matched to at most one right-side vertex. Since $r$ is uniform in
$R\setminus V_1$, we have
\[
\Pr\!\left[(u,r)\in M_{V_1\cup\{r\},r}\mid V_1\right]
\le
\frac{1}{|R\setminus V_1|}
=
\frac{1}{(b-1)k+1}.
\]
Multiplying by the $k$ possible choices of the distinguished element of $V$
gives the claim.
\end{proof}

\begin{lemma}
\label{lemma:prob_sel_low}
Fix $u\in L$ and $t\in\{T,\ldots,m\}$. Let
\[
A_t(u)
:=
\{u\notin e_{j,s}\text{ for all }j\in\{T,\ldots,t-1\},\ s\in[k]\}.
\]
Conditioning on the future blocks $(R_t,\ldots,R_m)$,
\[
\Pr[A_t(u)\mid R_t,\ldots,R_m]
\ge
\frac{T-2}{t-2},
\]
with the convention that the ratio is $1$ when $t=T$. In particular, when
$T=\Theta(m)$,
\[
\Pr[A_t(u)\mid R_t,\ldots,R_m]
\ge
\frac{T-1}{t-1}-O\!\left(\frac1m\right).
\]
\end{lemma}

\begin{proof}
The claim is trivial for $t=T$, since no block in
$\{T,\ldots,t-1\}$ has arrived.

For $b\ge T$, condition on the future blocks
$(R_{b+1},\ldots,R_m)$. The current block $R_b$ is a uniformly random
$k$-subset of the remaining $bk$ vertices. Applying
Lemma~\ref{lem:one_block_amortization} with this current block shows that
\[
\Pr[
u\in e_{b,s}\text{ for some }s\in[k]
\mid R_{b+1},\ldots,R_m
]
\le
\frac{k}{(b-1)k+1}
\le
\frac1{b-1}.
\]
Hence
\[
\Pr[
u\notin e_{b,s}\text{ for all }s\in[k]
\mid R_{b+1},\ldots,R_m
]
\ge
1-\frac1{b-1}
=
\frac{b-2}{b-1}.
\]
Iterating this bound for blocks $b=T,T+1,\ldots,t-1$ yields
\[
\Pr[A_t(u)\mid R_t,\ldots,R_m]
\ge
\prod_{b=T}^{t-1}\frac{b-2}{b-1}
=
\frac{T-2}{t-2}.
\]
Finally,
\[
\frac{T-2}{t-2}
=
\frac{T-1}{t-1}-O\!\left(\frac1m\right)
\]
uniformly over $t\in[T,m]$ when $T=\Theta(m)$.
\end{proof}

The preceding lemma controls tentative edges. The next lemma controls the
fallback/random matches that are used either in the initial blocks or when
the selected tentative endpoint is unavailable.

\begin{lemma}
\label{lemma:prob_random_match}
Fix $u\in L$ and $t\in[T,m]$. Conditional on the event $A_t(u)$ and on
$(R_t,\ldots,R_m)$,
\[
\Pr[
u \text{ is not used before block }t
\mid A_t(u), R_t,\ldots,R_m
]
\ge
1-y_k\frac{t-1}{m}.
\]
\end{lemma}

\begin{proof}
Let $H_t$ be the total number of online vertices matched before block $t$.
Each earlier block produces a match with probability $y_k$: in the first
$T-1$ blocks this is by construction, and in later blocks it is the hire
probability of the IC secretary subroutine. This remains true conditional
on any fixed ordered partition into unordered blocks: the candidate values
in each block are then fixed, while their order is uniform and the secretary
coins are independent. The event $A_t(u)$ depends only on these unordered
blocks. Consequently,
\[
\E[H_t\mid A_t(u),R_t,\ldots,R_m]=y_k(t-1).
\]

Recall that fallback matches are implemented using a uniformly random
priority order $\rho$ over $L$: whenever the algorithm needs a fallback
offline vertex, it takes the highest-ranked available vertex under $\rho$.
On the event $A_t(u)$, the vertex $u$ is never used as a tentative endpoint
before block $t$. Therefore, if $u$ is nevertheless matched before block $t$,
it must be because the fallback rule reaches $u$. This can happen only if
$u$ lies among the first $H_t$ vertices in the random priority order $\rho$:
if $u$ has rank larger than $H_t$, then fewer than $\rank_\rho(u)$ offline
vertices have been matched before block $t$, so the fallback rule cannot
have selected $u$.

Since $\rho$ is uniform and independent of the block partition and secretary
randomness,
\[
\Pr[
u \text{ is matched before block }t
\mid A_t(u), R_t,\ldots,R_m, H_t
]
\le
\frac{H_t}{m}.
\]
Taking expectations gives
\[
\Pr[
u \text{ is matched before block }t
\mid A_t(u), R_t,\ldots,R_m
]
\le
\frac{\E[H_t\mid A_t(u),R_t,\ldots,R_m]}{m}
\le
y_k\frac{t-1}{m}.
\]
Taking complements proves the claim.
\end{proof}

\begin{proof}[Proof of Lemma~\ref{lemma:amortized_availability}]
Fix $u\in L$ and $t\in[T,m]$. For $u$ to be available at the beginning of
block $t$, it suffices that it has not appeared as a tentative endpoint in
blocks $T,\ldots,t-1$ and has not been consumed by a fallback/random match
before block $t$. Lemmas~\ref{lemma:prob_sel_low} and
\ref{lemma:prob_random_match} therefore imply
\[
\Pr[
u\text{ is available at the beginning of block }t
\mid R_t,\ldots,R_m
]
\ge
\frac{T-2}{t-2}
\left(1-y_k\frac{t-1}{m}\right).
\]
Since $T=\Theta(m)$ in the optimized algorithm,
\[
\frac{T-2}{t-2}
=
\frac{T-1}{t-1}
-
O\!\left(\frac1m\right),
\]
and hence
\[
\Pr[
u\text{ is available at the beginning of block }t
\mid R_t,\ldots,R_m
]
\ge
\left(
\frac{T-1}{t-1}
-
O\!\left(\frac1m\right)
\right)
\left(1-y_k\frac{t-1}{m}\right).
\]
The $O(1/m)$ loss is dominated by the final
$O(1/\sqrt m)$ error term in Theorem~\ref{thm:ic_edge_wt2}.
\end{proof}

\subsection{Riemann-sum estimate}

Finally, we record the standard discretization estimate used in the proof of
Theorem~\ref{thm:ic_edge_wt2}. Let $T=\lceil \alpha m\rceil$ for a constant
$\alpha\in(0,1)$ and let $\gamma\in[0,1]$. Then
\[
\frac1m
\sum_{t=T}^m
\frac{T-1}{t-1}
\left(1-\gamma\frac{t-1}{m}\right)
=
\alpha
\int_\alpha^1
\frac{1-\gamma x}{x}\,dx
+
O\!\left(\frac1m\right).
\]
The integral equals
\[
\alpha
\int_\alpha^1
\frac{1-\gamma x}{x}\,dx
=
\alpha
\left(
\log\frac1\alpha
-
\gamma(1-\alpha)
\right).
\]
Thus the objective optimized in Section~\ref{sec:edge_wt} is
\[
f(\alpha;\gamma)
=
\alpha
\left(
\log\frac1\alpha
-
\gamma(1-\alpha)
\right).
\]

\section{Details for Binary-Weighted Complete Graphs}
\label{app:binary}

We give the full algorithm and proof of Theorem~\ref{thm:unweighted}. The
algorithm below is the binary-weighted block-greedy algorithm described in
Section~\ref{sec:unweighted}.

\begin{algorithm}
\caption{Block-Greedy IC Algorithm for Binary Weights}
\label{alg:block_matching_binary}

\KwIn{Offline vertices $L$ with $|L|=m$, online vertices $R$ with $|R|=km$}

Process the arrival sequence in $m$ consecutive blocks
$R_1,\ldots,R_m$, each of size $k$\;
$M_0\gets \emptyset$\;

\For{$t=1$ \KwTo $m$}{
    $M_t\gets M_{t-1}$\;
    Initiate a fresh IC secretary instance of Algorithm~\ref{alg:ic_sec}
    with $k$ candidates and the optimal parameter choice\;

    \For{$s=1$ \KwTo $k$}{
        Observe the arriving vertex $r_{t,s}$\;

        \If{some vertex in $R_t$ has already been matched}{
            \textbf{break}\;
        }

        \eIf{$r_{t,s}$ has an unmatched neighbor through a weight-$1$ edge}{
            $\gain(r_{t,s})\gets 1$\;
        }{
            $\gain(r_{t,s})\gets 0$\;
        }

        Feed value $\gain(r_{t,s})$ to the IC secretary subroutine,
        breaking ties by an independent infinitesimal perturbation or a
        fixed random tie-breaking rule\;

        \If{the IC secretary subroutine selects $r_{t,s}$}{
            \eIf{$\gain(r_{t,s})=1$}{
                Choose an available weight-$1$ neighbor $\ell_{t,s}$ of
                $r_{t,s}$\;
            }{
                Choose $\ell_{t,s}$ uniformly at random from the available
                offline vertices\;
            }
            $M_t\gets M_{t-1}\cup\{(\ell_{t,s},r_{t,s})\}$\;
        }
    }
}
\Return{$M_m$}\;
\end{algorithm}

We first show that the block constraint loses only a constant factor of the
offline optimum.

\begin{lemma}
\label{lem:block_opt_binary}
Let $M^\star$ be an optimal offline matching. Let $\pi$ be the uniformly
random arrival permutation, and let $M^\star_\pi$ be the maximum-weight
matching subject to matching at most one online vertex from each block. Then
\[
\mathbb{E}_\pi[w(M^\star_\pi)]\ge \left(1-\frac1e\right)w(M^\star).
\]
\end{lemma}

\begin{proof}
Let $M=w(M^\star)$, and mark the $M$ online vertices that are matched by
weight-$1$ edges in $M^\star$. For any fixed block $R_j$,
\[
\Pr_\pi[R_j \text{ contains no marked vertex}]
=
\frac{\binom{|R|-M}{k}}{\binom{|R|}{k}}
\le
\left(1-\frac{M}{|R|}\right)^k
=
\left(1-\frac{M}{km}\right)^k .
\]
Since $M\le m$,
\[
\Pr_\pi[R_j \text{ contains a marked vertex}]
\ge
1-e^{-M/m}
\ge
\left(1-\frac1e\right)\frac{M}{m}.
\]
Thus the expected number of blocks containing at least one marked vertex is
at least $(1-1/e)M$. Choosing one marked edge from each such block gives a
feasible matching that uses at most one online vertex per block, because the
marked edges come from the matching $M^\star$ and hence have distinct offline
endpoints. Therefore
\[
\mathbb{E}_\pi[w(M^\star_\pi)]\ge \left(1-\frac1e\right)M.
\]
\end{proof}

\begin{proof}[Proof of Theorem~\ref{thm:unweighted}]
The incentive-compatibility claim follows directly from the block structure.
In every block, the algorithm runs the same incentive-compatible secretary
subroutine, and every secretary selected by this subroutine is matched. Hence
all positions inside a block have the same ex-ante match probability; since
every block uses the same rule, the match probability is the same across all
arrival positions.

It remains to prove the competitive ratio. Fix an unordered partition into
blocks and choose a constrained optimal matching $M^\star_\pi$ independently
of their eventual order. Order the blocks uniformly at random and independently
permute the vertices inside each block. For any block contributing a weight-$1$
edge to $M^\star_\pi$, condition on its arrival position $j$ and write its edge as
$(\ell_j^\star,r_j^\star)$.

We use the following accounting scheme. Whenever the algorithm matches a
weight-$1$ edge $(\ell,r)$ in block $R_j$, assign reward $\beta$ to the
offline vertex $\ell$ and reward $1-\beta$ to the block $R_j$. Each
weight-$1$ edge therefore contributes total accounting reward exactly $1$.

Fix a contributing block $R_j$ and its offline endpoint $\ell_j^\star$. Let
\[
\alpha_j
:=
\Pr[\ell_j^\star \text{ is already matched by a weight-}1\text{ edge when }
R_j \text{ arrives}],
\]
and
\[
\gamma_j
:=
\Pr[\ell_j^\star \text{ is already matched by a weight-}0\text{ edge when }
R_j \text{ arrives}].
\]
If $\ell_j^\star$ is already matched by a weight-$1$ edge, the offline
endpoint receives accounting reward $\beta$. If $\ell_j^\star$ is unmatched,
then the block contains a candidate of gain $1$, and the IC secretary
subroutine selects a gain-$1$ candidate with probability at least $q_k$.
Therefore the expected combined accounting reward of the pair
$(R_j,\ell_j^\star)$ is at least
\[
\alpha_j\beta
+
(1-\alpha_j-\gamma_j)q_k(1-\beta).
\]
The only remaining bad case is that $\ell_j^\star$ was consumed earlier by a
weight-$0$ match. Since the algorithm makes at most one match in each earlier
block, and weight-$0$ fallback matches choose their offline endpoint
uniformly from the available vertices, we use the crude bound
\[
\gamma_j\le \frac{j-1}{m}.
\]
Hence
\begin{align*}
\mathbb{E}[\text{Reward}(R_j)+\text{Reward}(\ell_j^\star)]
&\ge
\alpha_j\beta
+
\left(1-\alpha_j-\frac{j-1}{m}\right)q_k(1-\beta)\\
&=
\alpha_j\bigl(\beta-q_k(1-\beta)\bigr)
+
\left(1-\frac{j-1}{m}\right)q_k(1-\beta).
\end{align*}
Choose
\[
\beta=\frac{q_k}{1+q_k}.
\]
Then $\beta=q_k(1-\beta)$, so the dependence on $\alpha_j$ cancels and
\[
\mathbb{E}[\text{Reward}(R_j)+\text{Reward}(\ell_j^\star)]
\ge
\left(1-\frac{j-1}{m}\right)\frac{q_k}{1+q_k}.
\]
Each contributing block has a uniformly random arrival position. Averaging its
bound over $j=1,\ldots,m$ gives the factor
\[
\frac1m\sum_{j=1}^m
\left(1-\frac{j-1}{m}\right)
=
\frac{m+1}{2m}.
\]
Summing over the contributing blocks, conditioned on the unordered partition
and averaging over block order and internal randomness, therefore gives
\[
\mathbb{E}[\text{weight obtained by Algorithm~\ref{alg:block_matching_binary}}]
\ge
\frac{m+1}{2m}\cdot \frac{q_k}{1+q_k}\cdot w(M^\star_\pi).
\]
Taking expectation over the random partition and applying
Lemma~\ref{lem:block_opt_binary},
\[
\mathbb{E}[\text{ALG}]
\ge
\left(1-\frac1e\right)
\frac{m+1}{2m}\cdot \frac{q_k}{1+q_k}\cdot w(M^\star).
\]
This proves the claimed competitive ratio.
\end{proof}

\section{Padding for Arbitrary Arrival Counts}
\label{app:padding}

In this appendix we extend the arrival-time IC property of
Algorithms~\ref{alg:ic_edge_wt} and~\ref{alg:block_matching_binary} to arbitrary
numbers of online and offline vertices. The main idea is to place each dummy
vertex at the end of its block, so that the real vertices in the block form a
random-order prefix of the secretary instance.

Let $m,n\ge 1$, $k:=\lceil n/m\rceil$, and $d:=km-n$. Before arrivals begin,
independently sample a uniformly random subset $D\subseteq[m]$ of cardinality
$d$. This is possible since $0\le d<m$. For every $t\in D$, block $t$ consists
of the next $k-1$ real arrivals followed by one dummy vertex. For every
$t\notin D$, block $t$ consists of the next $k$ real arrivals. Thus every block
has $k$ vertices, and the total number of real vertices is
$(m-d)k+d(k-1)=n$. Each dummy has weight $0$ to every offline vertex. We run
the algorithm on these $m$ blocks, processing each dummy according to the
same rule as the other vertices. An offline vertex matched to a dummy remains
unavailable for subsequent matches. At the end, we remove the dummy edges
from the output matching.

This construction can be implemented online. The number of real arrivals in
each block is fixed before arrivals begin, and each real vertex is processed
immediately on arrival. A dummy is processed after the prescribed number of
real arrivals in its block. When $k=1$, some blocks consist only of a dummy.

Let $\eta_k:=y_k/k$ be the parameter of the $k$-candidate IC secretary
subroutine, and use this same parameter in every secretary instance. All
secretary coins and tie-breaking priorities are sampled independently of the
arrival order and of $D$, with fresh randomness for each block.

\begin{proposition}[Arrival-time IC under padding]
\label{prop:padding_ic}
Under the construction above, Algorithms~\ref{alg:ic_edge_wt}
and~\ref{alg:block_matching_binary} are arrival-time IC for the original
online vertices. In either algorithm, for every real arrival position
$i\in[n]$,
\[
\Pr[r_i\text{ is matched}]=\eta_k.
\]
The equality also holds conditional on any fixed choice of $D$.
\end{proposition}

We first record a property of the IC secretary algorithm. Its selection
probability is unchanged if we run it only on a random-order prefix.

\begin{lemma}[Selection probabilities on a prefix]
\label{lem:secretary_prefix_ic}
Fix $k\ge 1$, $0<\eta\le 1/k$, and $0\le h\le k$. Let $h$ candidates with
fixed values arrive in a uniformly random order, with ties broken by
priorities independent of arrival order. Run the first $h$ steps of the
$k$-candidate mechanism $\mathcal{M}_\eta$ in Algorithm~\ref{alg:ic_sec}.
Then every position $s\in[h]$ is selected with probability exactly $\eta$,
and the probability of making a selection in these $h$ steps is $h\eta$.
\end{lemma}

\begin{proof}
The claim is immediate when $h=0$. Otherwise, condition on the tie-breaking
priorities, so that the candidates have a fixed strict ordering. Under a
uniformly random permutation, the relative rank at position $s$ is uniform
on $[s]$ and independent of the preceding relative ranks. The event of no
earlier selection depends only on these preceding relative ranks and the
independent secretary coins.

Write
\[
a_s:=\frac{s}{1/\eta-s+1}.
\]
If $a_s\le 1$, the mechanism selects the current candidate with probability
$a_s$ when it is the best so far. If $a_s>1$, it selects any of the top
$\lfloor a_s\rfloor$ relative ranks, and selects relative rank
$\lfloor a_s\rfloor+1$ with probability $a_s-\lfloor a_s\rfloor$.
Since $s\le k$ and $\eta\le 1/k$, we have $a_s\le s$. Therefore, in either
case,
\[
\Pr[\text{selection at }s\mid\text{no earlier selection}]
=\frac{a_s}{s}
=\frac{\eta}{1-(s-1)\eta}.
\]
We now proceed by induction. If each preceding position is selected with
probability $\eta$, then, since at most one candidate is selected, the
probability of no selection before position $s$ is $1-(s-1)\eta$. Hence
\[
\Pr[\text{selection at }s]
=\bigl(1-(s-1)\eta\bigr)
\frac{\eta}{1-(s-1)\eta}
=\eta.
\]
The same calculation gives the base case $s=1$. Summing over the $h$
positions proves the final claim.
\end{proof}

\begin{proof}[Proof of Proposition~\ref{prop:padding_ic}]
Fix $D$. Let $h_t:=k-\mathbf{1}_{\{t\in D\}}$ be the number of real
vertices in block $t$. These block sizes are fixed independently of the
real arrival order. Condition on the history before block $t$ and on the
unordered set of real vertices assigned to this block. Their order within
the block remains uniformly random.

For Algorithm~\ref{alg:ic_edge_wt}, the tentative value of each real vertex
is computed using the preceding blocks together with that vertex alone.
Thus the real candidate values form a fixed multiset, independent of their
order within the block. For $t\ge T$, Lemma~\ref{lem:secretary_prefix_ic}
shows that each of the first $h_t$ positions is selected with probability
$\eta_k$. Any dummy is processed afterward and cannot affect these earlier
decisions. For $t<T$, the algorithm chooses a uniformly random position among
the $k$ padded positions with probability $y_k$. Each real position is
therefore again selected with probability $y_k/k=\eta_k$.

For Algorithm~\ref{alg:block_matching_binary}, the set of available offline
vertices is fixed until the first selection in the block. We may therefore
evaluate the gain of every real candidate against the available set at the
beginning of the block. The algorithm sees these fixed gains in a uniformly
random order until it makes a selection, after which its secretary instance
terminates. Lemma~\ref{lem:secretary_prefix_ic} again implies selection
probability $\eta_k$ for every real position.

In either algorithm, at most one offline vertex is consumed per block,
including matches to dummies. Before a selection in block $t$, at least
$m-t+1\ge 1$ offline vertices remain available. Since the graph is complete,
every selected real vertex is matched, using a fallback edge if necessary.

For fixed $D$, each original arrival index $i\in[n]$ belongs to a fixed
block and a fixed position in its real prefix. Taking expectations over the
conditioning above gives
\[
\Pr[r_i\text{ is matched}\mid D]=\eta_k
\qquad\text{for every }i\in[n].
\]
Averaging over $D$ proves the proposition. Blocks with $h_t=0$ contain no
real arrival positions and require no additional argument.
\end{proof}

The construction above preserves arrival-time IC, but the padded sequence
is not a uniformly random permutation of all real and dummy vertices.
Our competitive-ratio guarantees in
Theorems~\ref{thm:ic_edge_wt} and~\ref{thm:unweighted} concern $n=km$;
extending these guarantees to arbitrary $n$ requires a separate analysis.

\end{document}